\documentclass[twocolumn,english]{IEEEtran}
\usepackage[T1]{fontenc}
\usepackage{amssymb,amsmath, graphicx, psfrag, hyperref, graphicx}
\pdfoutput=1

\makeatletter

\newtheorem{thm}{Theorem}[section]

\newtheorem{lem}[thm]{Lemma}

\usepackage{babel}
\makeatother
\begin{document}

\title{Distributed Joint Source-Channel Coding  for
arbitrary memoryless correlated sources and Source coding for
Markov correlated sources using LDPC codes}

\author{Vaneet Aggarwal\\Princeton University,
Princeton, NJ 08544\\Email: vaggarwa@princeton.edu}

\maketitle
\begin{abstract}
In this paper, we give a distributed joint source channel coding
scheme for arbitrary correlated sources for arbitrary point in the
Slepian-Wolf rate region, and arbitrary link capacities using LDPC
codes. We consider the Slepian-Wolf setting of two sources and one
destination, with one of the sources derived from the other source
by some correlation model known at the decoder. Distributed
encoding and separate decoding is used for the two sources. We
also give a distributed source coding scheme when the source
correlation has memory to achieve any point in the Slepian-Wolf
rate achievable region. In this setting, we perform separate
encoding but joint decoding.
\end{abstract}

\section{Introduction}

\PARstart{D}{istributed} source coding (DSC) refers to the
compression of the outputs of two or more physically separated
sources that do not communicate with each other (hence the term
distributed coding). These sources send their compressed outputs
to a central point (e.g., the base station) for joint decoding.
Driven by a host of emerging applications (e.g., sensor networks
and wireless video), DSC has recently become a very active
research area - more than 30 years after Slepian and Wolf laid its
theoretical foundation \cite{sw}.

Wyner first realized the close connection of DSC to channel coding
and suggested the use of linear channel codes as a constructive
approach for Slepian-Wolf coding in his 1974 paper \cite{wyn}.
Wyner's scheme was only recently used in \cite{discus} for
practical Slepian-Wolf code designs based on conventional channel
codes like block and trellis codes. If the correlation between the
two sources can be modelled by a binary channel, Wyner's syndrome
concept can be extended to all binary linear codes; and
state-of-the-art near-capacity channel codes such as turbo
\cite{tur} and LDPC codes \cite{galla} \cite{urban} can be
employed to approach the Slepian-Wolf limit.

Slepian-Wolf rate region is the rate region in which there exist
encoding and decoding schemes such that reliable transmission can
take place, which was given in \cite{sw} as $R_1 \ge H(X|Y)$  ,
$R_2 \ge H(Y|X)$, and $R_1 +R_2 \ge H(X,Y)$, where $R_1$ and $R_2$
are the source encoding rates of the two sources X and Y, and
$H(X|Y)$ and $H(Y|X)$ the conditional entropies of the sources,
and $H(X,Y)$ the joint entropy of the two sources. The most
important points in the Slepian-Wolf rate region are the
asymmetric point ( in which one source, say X, is encoded at the
rate that equals its unconditional entropy (H(X) ) and the other
source, say Y, is encoded at the rate equal to its conditional
entropy ($H(Y|X)$ )), and the symmetric point (in which both
sources are encoded at the same rate $H(X,Y)/2$). Constructive
approaches (e.g., \cite{appsen}) have been proposed for the
symmetric point of the boundary of the Slepian-Wolf rate region,
for the asymmetric point of the boundary of the Slepian-Wolf rate
region (example \cite{discus}\cite{ct}\cite{cld} \cite{vernoi})
and also for any point between on the boundary (for example
\cite{dca}\cite{sarfek} ).  Most of the cases in the literature
deal with correlation model being a binary symmetric channel,
though some work on arbitrary correlation has also been done (for
example \cite{dca}).

After the practical scheme of \cite{discus}, there was a lot of
work in the direction of the practical code construction due to
its application in sensor networks \cite{appsen}. Some of the
works relating to this field are
\cite{dca}\cite{jscn}\cite{ct}\cite{tsc}\cite{dst}\cite{cbt}\cite{ldpa}\cite{cld}.
In \cite{dca}, practical construction using syndrome approach was
given for arbitrarily correlated sources. There is also another
scheme not using the syndrome approach, but rather encoding in the
same way as the channel coding, by multiplying at the encoder by
the generator matrix as is given in \cite{sarfek}( The idea of
using such approach is also there in \cite{swcoop}). They consider
in the paper source correlation model as BSC with error
probability p. To encode, a non-uniform systematic LDPC code with
generator matrix of size n $\times$ n(1+h(p)) is used (where h(p)
is the binary entropy of p), and at the encoder of the first
source, say $X_1$, 1/a fraction of the information bits, and
(1-1/a) fraction of the parity bits are sent. At the encoder of
the other source, say $X_2$,  the remaining (1-1/a) fraction of
information bits and the remaining 1/a fraction of the parity bits
are sent. Decoding is performed as follows: At decoder of $X_1$,
we have 1/a fraction of information bits of X with infinite
likelihood, (1-1/a) fraction of parity bits with infinite
likelihood, (1-1/a) fraction of information bits with ln((1-p)/p)
likelihood, from which we decode $X_1$, and similarly we can
decode $X_2$. Note that the remaining parity bits from $X_2$ are
not used in the decoding of $X_1$, they are considered as being
punctured. They also suggest to use unequal error protection for
LDPC since the information bits are more important. More important
bits have higher degrees while the rest have lower degree. Hence,
they design systematic LDPC code with distinct variable node
degree distributions for information and parity bits. Our scheme
will extend upon this basic idea of doing source coding in the
same way as traditional channel coding scheme.

There has also been some work in the direction of joint source
channel decoding of correlated sources over noisy channels (for
example \cite{jscn}, \cite{live}, \cite{fff}, \cite{newsc}, and
\cite{nearsh}) but none of the work till now covers the problem in
that generality using LDPC codes. It has been proven in
\cite{ldpa} that LDPC codes can approach the Slepian Wolf bound
for general binary sources, and hence the choice of LDPC codes
\cite{galla}\cite{urban} in this work.

Most of the cases in the literature deal with correlation model
being a binary symmetric channel, though some work on arbitrary
correlation has also been done (for example \cite{dca}) and also
for correlation having Markov structure (for example \cite{markc}
\cite{vernoi}). To the best of our knowledge, no work for
arbitrary point in the Slepian-Wolf rate region is there in the
literature although there are attempts for asymmetric point in the
rate region( for example \cite{markc} \cite{vernoi}).

The paper is organized as follows: We will start with giving the
scheme of joint source channel coding for the source correlation
model being being arbitrary memoryless with arbitrary link
capacities at arbitrary point in the Slepian-Wolf Rate region in
Section \ref{bscmodel}. Then, we give some simulation results for
this scheme in Section \ref{simuarb}. Following this, we give a
scheme for doing distributed source coding at arbitrary rates in
the Slepian-Wolf Rate Achievable region when the source
correlation has memory in Section \ref{marcor} and some
simulations for the scheme are given in Section \ref{simmar}.
Finally, we give some concluding remarks in Section \ref{conc}.
\section{Joint Source Channel Coding Scheme for Arbitrary Source Correlation}
\label{bscmodel}

Consider first that we encode $X_1$ as follows: $X_1$ is fed into
systematic LDPC code of Rate $R_1$. The encoder sends the
corresponding parity bits and a 1/a fraction of the information
bits. These parity bits involve the bits needed for source as well
as channel. Let $P_1$ be the parity bits needed for source and
$P_3$ the parity bits needed for the channel for source $X_1$. The
compression rate of source $X_1$ is $ R_{X_1 } = \frac{{k/a + P_1
}} {k}$, and the corresponding LDPC code has rate $R_1 = \frac{k}
{{k + P_1  + P_3 }}$ where k is the number of bits of $X_1$ and of
$X_2$ at the input. The same procedure with a few modifications is
applied at the source $X_2$. For this source, we use a systematic
LDPC code of rate $R_2$. The encoder sends the related parity bits
and the remaining 1-1/a fraction of the information bits.
Similarly, assume that the parity bits are $P_2$ for the source
coding and $P_4$ for the channel coding at the source $X_2$, we
have the compression rate and the LDPC code rate equal to $ R_{X_2
} = \frac{{(a-1)k/a + P_2 }} {k}$ and $R_2 = \frac{k} {{k + P_2  +
P_4 }}$ respectively.

The procedure as we discussed above involves two LDPC codes.
However, we are interested in designing a single channel code for
both sources. So, we will make a single LDPC matrix having the
information bits, and all parity bits corresponding to $P_1$,
$P_2$, $P_3$, and $P_4$.
\begin{lem} $P_1$ + $P_2$ = k($R_{X_1 }$ + $R_{X_2 }$ -1 )\end{lem}
\begin{proof}
Add $R_{X_1 }$ and $R_{X_2 }$ given above to get the required

 \end{proof}

\begin{lem} $ P_3  + P_4 = k [R_{X_1 } (\frac{1}
{{R_{c_1} }} - 1) + R_{X_2 } (\frac{1} {{R_{c_2} }} - 1)]$
\end{lem}
\begin{proof}
\[
\begin{gathered}
  R_{c_1}  = \frac{{kR_{X_1 } }}
{{kR_{X_1 }  + P_3 }} \hfill\\ R_{c_2}  = \frac{{kR_{X_2 } }}
{{kR_{X_2 }  + P_4 }} \hfill \\
   \Rightarrow P_3  = kR_{X_1 } (\frac{1}
{{R_{c_1} }} - 1),P_4  = kR_{X_2 } (\frac{1}
{{R_{c_2} }} - 1) \hfill \\
\end{gathered}
\]
Add $P_3$ and $P_4$ to get the required  result.

 \end{proof}

Let the rate of the LDPC code needed is $R$. With this parity
check matrix, we will generate required parity bits for each
source. Now, as we can see, the modified algorithm is to encode
both the sources with a single LDPC code, and at encoder of $X_1$,
send the first fraction 1/a of the information bits, and the first
1/b fraction of the parity bits, and at encoder of $X_2$, send the
remaining 1-1/a information bits and last 1/c fraction of the
parity bits (Figure \ref{ena1}). Also keep in mind that  to
satisfy the Slepian-Wolf Compression bounds, $R_{X_1}$ and
${R_{X_2}}$ (the source compression rates of the two sources) have
to be in the Slepian-Wolf rate region. At the decoder of the each
source, we use the information bits coming from both the sources,
and the parity bits from that source, and the remaining bits are
considered as punctured. The parameters a,b,c,R, $R_{c_1}$,
$R_{c_2}$ (the channel rates of the two channels), $R_{X_1}$ and
${R_{X_2}}$ are related as in \eqref{eq1}. $C_f$ and $C_b$ are the
capacities of the channel from the source $X_1$ to $X_2$ and from
source $X_2$ to source $X_1$ respectively.

 \begin{equation}\label{eq1}
\begin{gathered}
  (1 - \frac{1} {c})(\frac{1} {R} - 1) =
  (1 - \frac{{R_{c_1 } }} {{R_{c_2 } }}C_f )(1 - \frac{1}
{a}) \hfill \\
  (1 - \frac{1} {b})(\frac{1} {R} - 1) = \frac{1} {a}(1 - \frac{{R_{c_2 } }} {{R_{c_1 } }} C_b) \hfill \\
  R_{X_1 }  = R_{c_1} [\frac{1} {a} + \frac{{\frac{1} {R} - 1}}
{b}] \hfill \\
  R_{X_2 }  = R_{c_2} [1- \frac{1} {a} + \frac{{\frac{1} {R} - 1}}
{c}] \hfill \\
\end{gathered}
\end{equation}

As we can see that punctured bits can be further decreased by
taking the total parity bits as the maximum number of parity bits
being used and which equals maximum of $P_1$ + $P_3$ and $P_2$ +
$P_4$. Hence we perform the second step  according to \eqref{eq2}.
\begin{equation}\label{eq2}
\begin{gathered}
  R' \leftarrow \frac{R} {{R + \max (\frac{{1 - R}} {b},\frac{{1 -
  R}}
{c})}} \hfill \\
  b \leftarrow b \times \frac{{R(1 - R')}}
{{R'(1 - R)}} \hfill \\
  c \leftarrow c \times \frac{{R(1 - R')}}
{{R'(1 - R)}} \hfill \\
R \leftarrow R' \hfill \\
\end{gathered}
\end{equation}

The encoding and decoding functions are explained in detail below:

$\\$Encoding:
    $   \\$
The scheme is explained in the figure \ref{ena1}.
\begin{figure}[h]
\includegraphics[width=9cm]{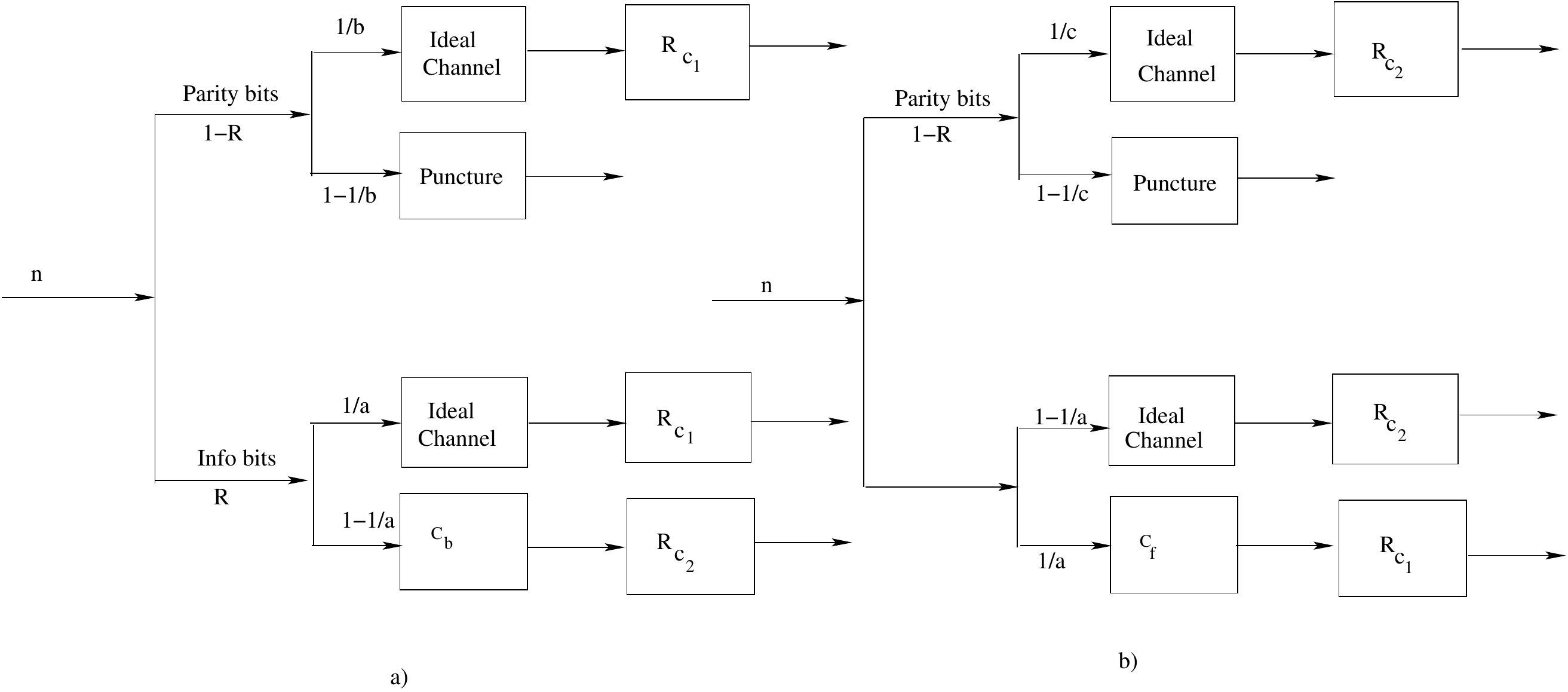}\\
 \caption{Parallel Channel for (a) source $X_1$ and (b) source $X_2$}
 \label{ena1}

\end{figure}
At the encoder of $X_1$, we send the first 1/a fraction of
information bits and first 1/b fraction of the parity bits. At the
encoder of source $X_2$, we send the remaining 1-1/a fraction of
information bits, and last 1/c fraction of parity bits. The LDPC
matrix made is of rate R . We have four variables here a,b,c,R for
given $R_{X_1 }$ , $R_{X_2 }$, $R_{c_1 }$ and $R_{c_2}$  for which
we use the set of equations \eqref{eq1} to solve them explicitly.
It can also be seen that
\[
\begin{gathered}
 R_1  = \frac{k} {{k + k(\frac{1} {R} - 1)( \frac{1} {b})}} =
\frac{1} {{1 + (\frac{1} {R} - 1)( \frac{1}
{b})}} \hfill \\
  R_2  = \frac{k}
{{k + k(\frac{1} {R} - 1)(\frac{1} {c})}} = \frac{1} {{1 +
(\frac{1} {R} - 1)(\frac{1}
{c})}} \hfill \\
\end{gathered}
\]

\begin{figure}[h]
\begin{center}
 \includegraphics[width=9cm]{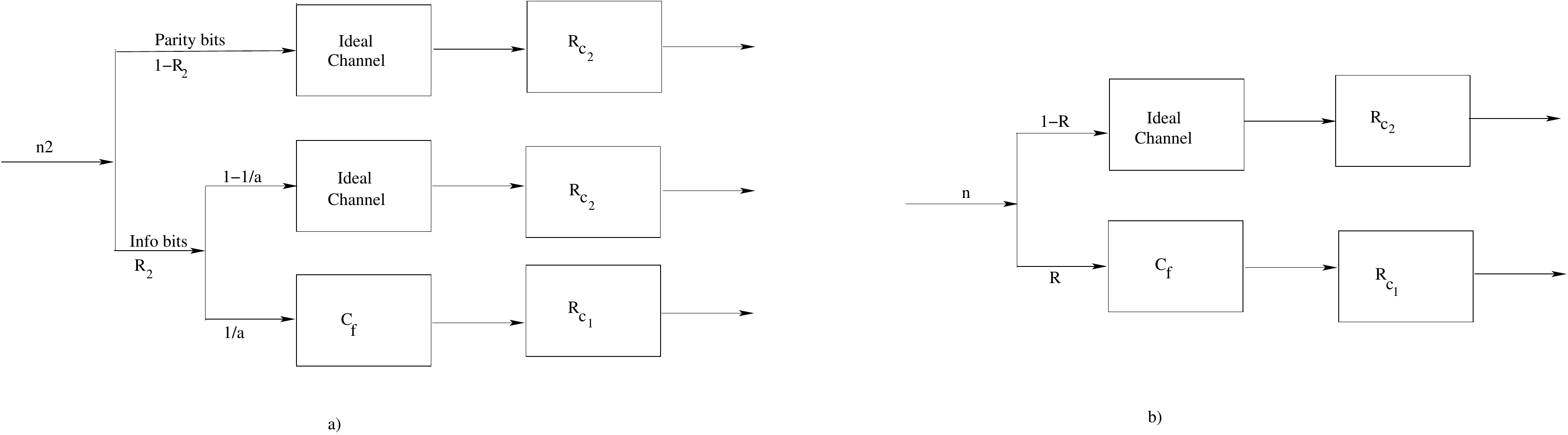}\\
  \caption{(a) Equivalent Model of Figure ~\ref{ena1}(b), and (b) Special case of this model as c=1 }
\label{ena2}
\end{center}
\end{figure}
\begin{figure}[h]
\begin{center}
 \includegraphics[width=9cm]{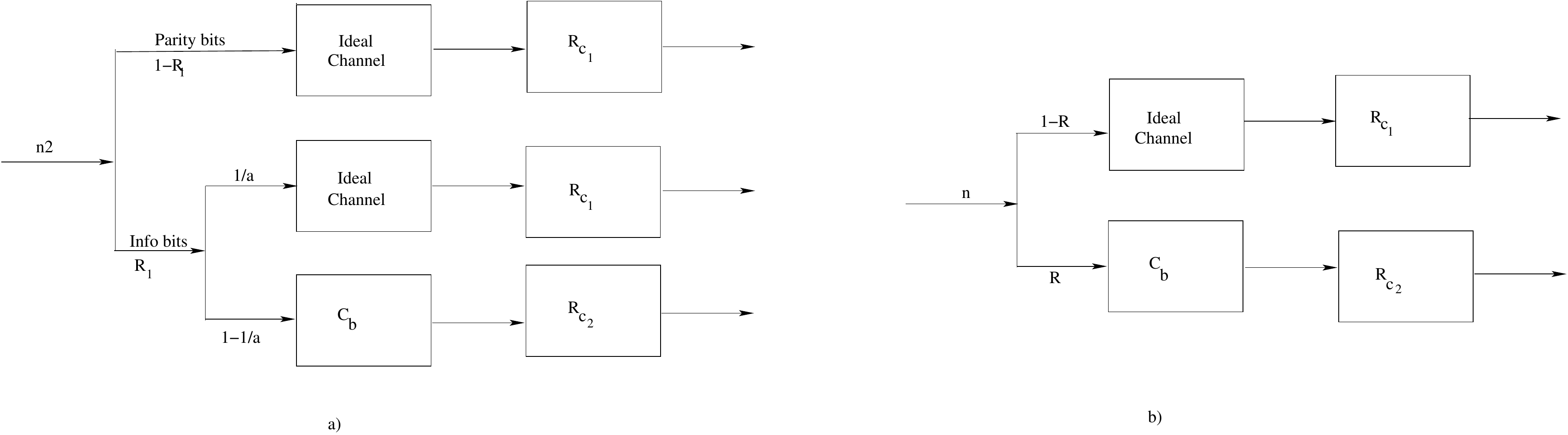}\\
  \caption{(a) Equivalent Model of Figure ~\ref{ena1}(a), and (b) Special case of this model as b=1 }
\label{ena21}
\end{center}
\end{figure}
Consider  channels of Figure \ref{ena1}(b). In Figure
~\ref{ena2}(a), we have an equivalent model with
$n_2$=n-(1-1/c)(n-k) , where k=nR, and the parity bits are $P_2$
and $P_4$. According to \cite{rc}, the performance of decoder do
not change after puncturing if R/C do not change. It is easy to
see using \eqref{eq1} that the model of  source $X_2$ is similar
to test model in Figure ~\ref{ena2}(b) when $c=1$. We will show
that the decoder performance of the model of  Figure
~\ref{ena1}(b) and figure ~\ref{ena2}(b) are the same in the next
few lemmas. Similarly, we will show that the decoder performance
of the model of  Figure ~\ref{ena1}(a) and figure ~\ref{ena21}(b)
(its equivalent model when b=1) which means that decoder
performance do not change with the choices of a,b,c.

\begin{lem} The ratio of rate and capacity are same for Figure
\ref{ena2}(a) (or Figure \ref{ena1}(b)) and Figure \ref{ena2}(b)
are the same
\end{lem}
\begin{proof}
As \[
\begin{gathered}
  C_2  = (1 - R_2 )R_{c_2} + R_2 R_{c_2} (1 - \frac{1}
{a}) + \frac{{R_2 }}
{a} C_f R_{c_1} \hfill \\
  R_2  = \frac{R}
{{R + (1 - R)( \frac{1}
{c})}} \hfill \\
  C = (1 - R)R_{c_2} + R C_f R_{c_1} \hfill \\
\end{gathered}
\]
After some manipulations, we get  $ \frac{{R_2 }} {{C_2 }} =
\frac{R} {C} $ using  Equations ~\eqref{eq1}.
\end{proof}

\begin{lem} The ratio of rate and capacity are same for Figure
\ref{ena1}(a) and Figure \ref{ena21}(b) are the same.

\end{lem}
\begin{proof}
As \[
\begin{gathered}
  C_1  = (1 - R_1 )R_{c_1} + R_1 (1 - \frac{1}
{a}) C_b R_{c_2} + R_{c_1} \frac{{R_1 }}
{a} \hfill \\
  R_1  = \frac{R}
{{R + (1 - R)(\frac{1}
{b})}} \hfill \\
  C = (1 - R)R_{c_1} + R C_b R_{c_2} \hfill \\
  \end{gathered}
  \]
   After some manipulations, we get $\frac{{R_1 }} {{C_1 }} =
\frac{R} {C}$ using Equations ~\eqref{eq1}.

\end{proof}

It is also easy to see that the bits sent satisfy equation
\eqref{eq1}, the total bits that are sent through encoder $X_1$
are\[ \frac{k} {a} + k(\frac{1} {R} - 1)(\frac{1} {b}) =
k\frac{{R_{X_1 } }} {{R_{c_1} }}
\] and the total bits that are sent through encoder $X_2$
are\[ k(1 - \frac{1} {a}) + k(\frac{1} {R} - 1)(\frac{1} {c}) =
k\frac{{R_{X_2 } }} {{R_{c_2} }}
\]

Hence all the equations in equation \eqref{eq1} makes good
sense.As we can see that punctured bits can be further decreased
by taking the total parity bits as the maximum number of parity
bits being used and which equals maximum of $P_1$ + $P_3$ and
$P_2$ + $P_4$. Hence we perform the second step  according to
equation \eqref{eq2}.

We still need to find $C_f$ and $C_b$ for solving the above
equations. Consider the forward model as in figure \ref{fwm}.
\begin{figure}[h]
\begin{center}
 \includegraphics[width=9cm]{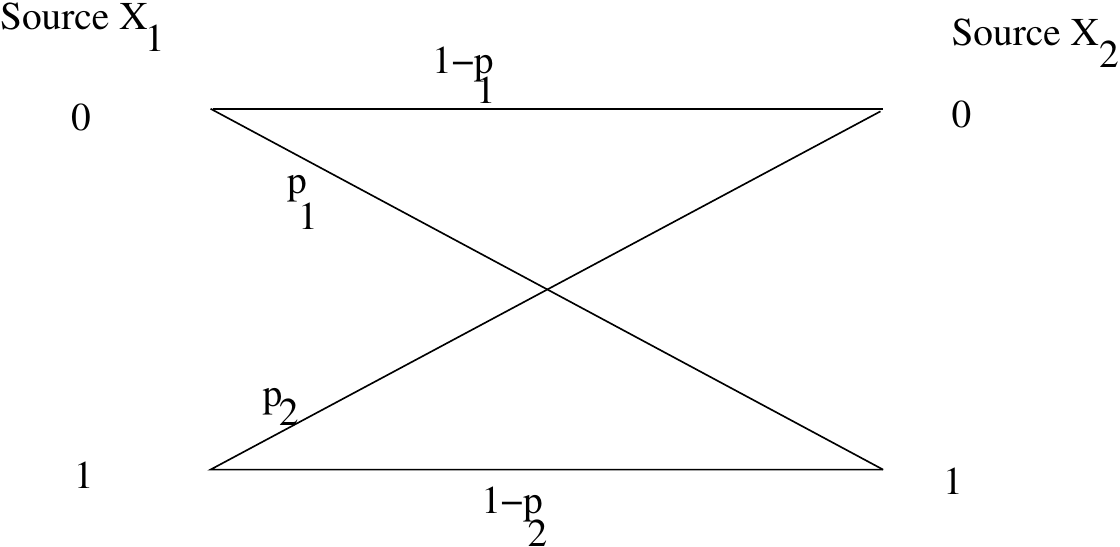}\\
  \caption{Forward Model}
\label{fwm}
\end{center}
\end{figure}
The capacity of the forward model is given by
\[
C_f  = h(\frac{1} {{1 + 2^{\frac{{h(p_1 ) - h(p_2 )}} {{1 - p_1  -
p_2 }}} }}) - \beta h(p_1 ) - (1 - \beta )h(p_2 )
\]

where $ \beta  = \frac{{\frac{1} {{1 + 2^{\frac{{h(p_1)  - h(p_2
)}} {{1 - p_1  - p_2 }}} }} - p_2 }} {{1 - p_1  - p_2 }} \\ $
 If the
probability that $(X_1 = 0) $ is $\alpha$, then the backward
channel is given as in figure \ref{bwm}
\begin{figure}[h]
\begin{center}
 \includegraphics[width=9cm]{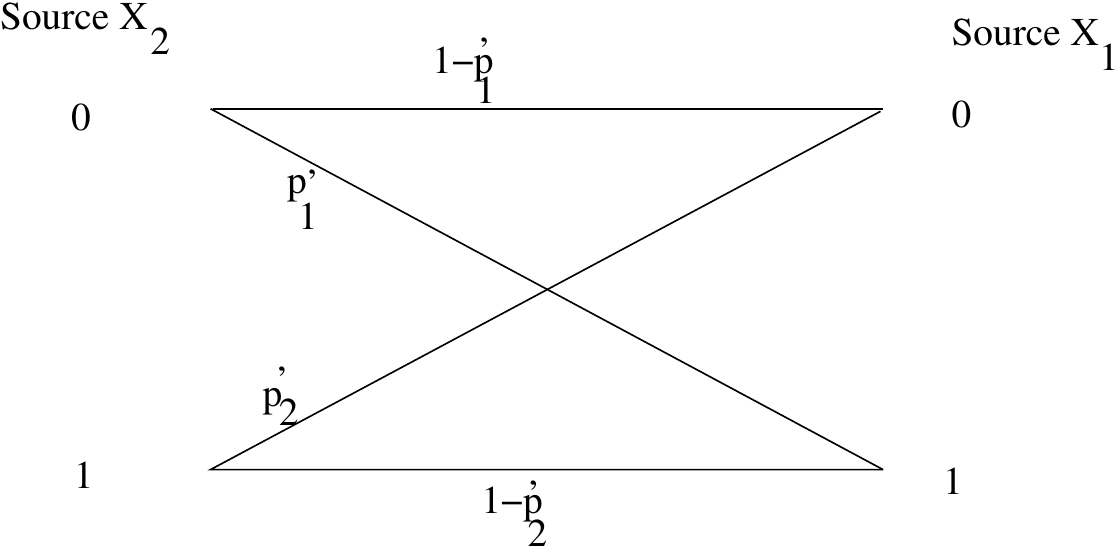}\\
  \caption{Backward Model}
\label{bwm}
\end{center}
\end{figure}

where $ p_1^\prime   = \frac{{p_2 (1 - \alpha )}} {{\alpha (1 -
p_1 ) + (1 - \alpha )p_2 }}$ and $ p_2^\prime  = \frac{{p_1
\alpha}} {{\alpha p_1 + (1 - \alpha)(1 - p_1 )}} $. Hence, the
capacity of the backward Channel is given by
\[
C_b  = h(\frac{1} {{1 + 2^{\frac{{h(p_1^\prime  ) - h(p_2^\prime
)}} {{1 - p_1^\prime   - p_2^\prime  }}} }}) - \gamma h(p_1^\prime
) - (1 - \gamma )h(p_2^\prime  )
\]

 where $
\gamma  = \frac{{\frac{1} {{1 + 2^{\frac{{h(p_1^\prime)   -
h(p_2^\prime  )}} {{1 - p_1^\prime   - p_2 }}} }} - p_2^\prime  }}
{{1 - p_1^\prime   - p_2^\prime  }}$

It is clear that I have given separate parameters for source and
channel encoding so that they can be separately decided, and then
use equations \eqref{eq1}, and then equations \eqref{eq2} to get
all the relevant parameters of the code and design.

According to \cite{rc},for any rates $R_1$ and $R_2$ that 0 <
$R_1$ < $R_2$ < 1, there exists an ensemble of LDPC codes with the
following property: The ensemble can be punctured from rate $R_1$
to $R_2$ resulting in asymptotically good codes for all rates $R_1
\le R \le R_2$. Hence, we design punctured codes according to
\cite{rc} for rate min($R_1$,$R_2$) and that will work for both
and this is the rate we get after the step 2 using the equations
\eqref{eq2}.

 $\\$Decoding:
    $   \\$
    The decoder needs to determine $X_1$ from information bits of
    $X_1$ (partly that were sent from encoder of $X_1$ and partly
    sent from encoder of $X_2$), and the parity bits $P_1$ and
    $P_3$. Likelihood of these parity bits are decided by the
    channel noise alone, and similar for the information bits
    coming from the encoder of $X_1$, while the information bits
    from the encoder of $X_2$ will have the effect of the BSC in
    the path from $X_1$ to $X_2$ also. Similar procedure will
    decode the source $X_2$.   $\\$Remark:
    $   \\$ The code
    construction(irregular LDPC codes) and decoding has to be done
    in the way suggested in \cite{mitz} which gives a scheme using
    density evolution with erasures and errors which is the model
    of our scheme. So, this model gives the density evolution
    analysis for this case, and we can choose the degree
    distribution and decoding parameters according to this density
    evolution analysis.

\section{Simulations of Joint source channel coding for arbitrary memoryless correlation}
\label{simuarb} As a first example, take $p_1$ = .05, $p_2 = .185$
and $\alpha$ =.4 in Figure \ref{fwm}. This gives a joint entropy
of two sources as 1.5. Solving $C_f$ and $C_b$ gives $C_f$ =
0.4998, $C_b$ = 0.4839. Also take the source rates $R_{X_1}$ = .8,
$R_{X_2}$ = .7, and the channel rates as $R_{c_1}$ = .9 and
$R_{c_2}$ = .94. Putting all these in equations \eqref{eq1} and
\eqref{eq2}, we get the parameters a,b, c and R as 1.58, 1.4754, 1
and .7259 respectively. These parameters should ideally work for
this scheme. To simulate, we use $p_1$ = .05 and $\alpha$ =.4 in
Figure \ref{fwm}, but we keep $p_2$ as a variable, changing which
we vary the joint entropy of the sources. Also, choose the both
the channels as BSC with error probabilities .0129 and .0069
respectively so that both the channels rates are the same as the
capacity. The simulation result is shown in Figure \ref{simu1}.
\begin{figure}[h]\begin{center}
 \includegraphics[width=9cm]{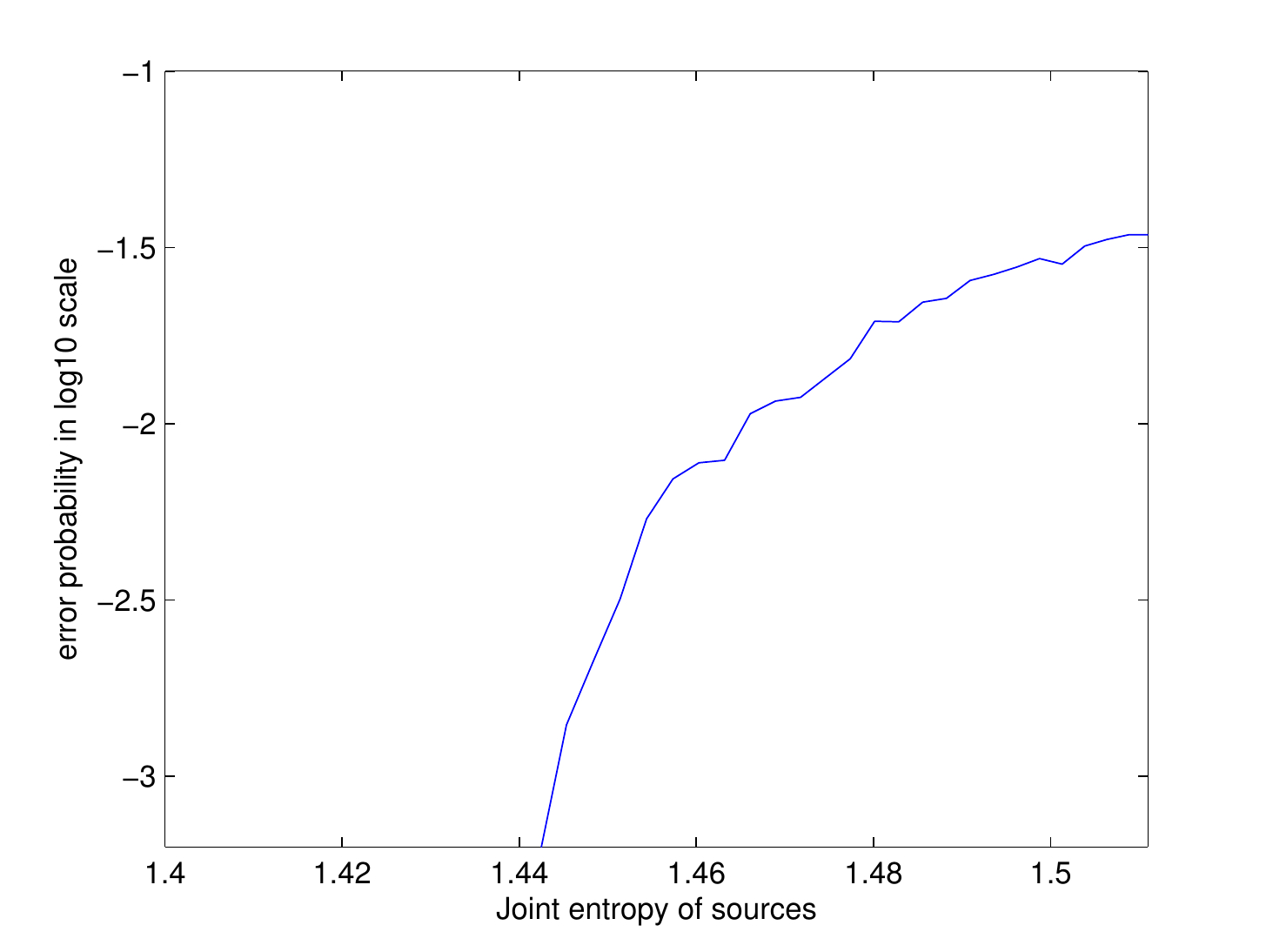}\\
 \caption{Simulation Results} \label{simu1}\end{center}
\end{figure}
\begin{figure}[h]\begin{center}
 \includegraphics[width=9cm]{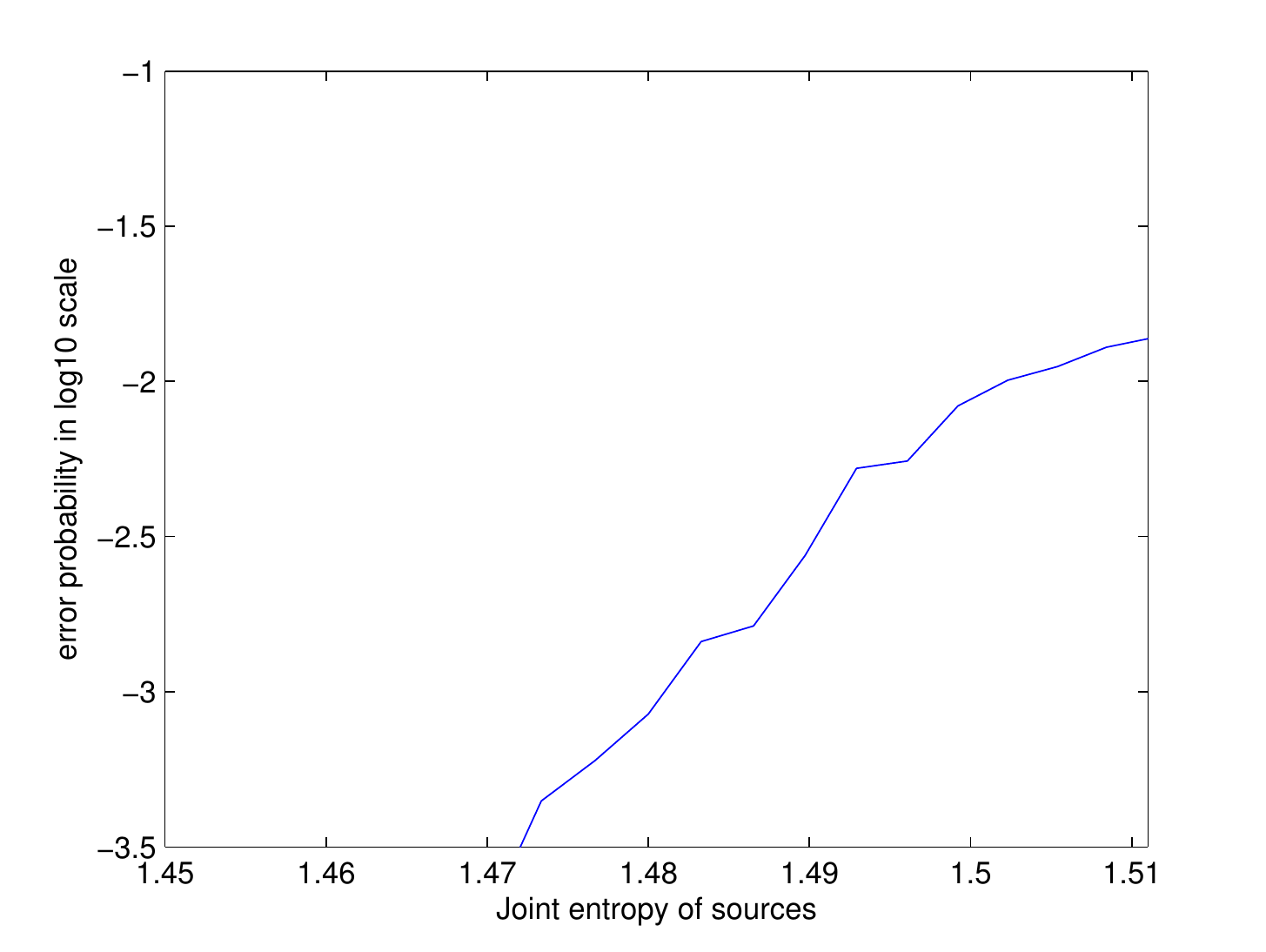}\\
 \caption{Simulation Results} \label{simu2}\end{center}
\end{figure}

As a second example, consider the model for $p_1$ = .09, $p_2 =
.143$ and $\alpha$ =.4. This gives a joint entropy of two sources
as 1.5. Solving $C_f$ and $C_b$ gives $C_f$ =  0.4839, $C_b$ =
0.4703. Also take the source rates $R_{X_1}$ = .8, $R_{X_2}$ = .7,
and the channel rates as $R_{c_1}$ = .85 and $R_{c_2}$ = .9.
Putting all these in equations \eqref{eq1} and \eqref{eq2}, we get
the parameters a,b, c and R as 1.5392, 1.4665, 1 and .7005
respectively. To simulate, we use $p_1$ = .09 and $\alpha$ =.4 in
Figure \ref{fwm}, but we keep $p_2$ as a variable, changing which
we vary the joint entropy of the sources. Also, choose the both
the channels as BSC with error probabilities .0215 and .0129
respectively so that both the channels rates are the same as the
capacity. The simulation result is shown in Figure \ref{simu2}.

\section{Distributed Source Coding Scheme for Markov Source Correlation}
\label{marcor}
 The correlation between the sources has a memory
defined by a  Markov model. The two source sequences be $X_1$ and
$X_2$ as before, with $X_2$ = $X_1$ + N, where N is generated by a
Markov model $ \lambda  = (A,B,\pi ) $ , and addition is modulo
two addition. The model is characterized by a set of states $S_j$,
$0 \le j \le S-1$, the matrix of transition probabilities among
states [A=($a_{ij}$), with $a_{ij} = P_t(S_j|S_i)$ the probability
of transition from state $S_i$ to $S_j$], and the list giving the
bit probability to associate with each state     [B=($b_{jv}$),
with $b_{jv} = P_0(v|S_j)$ the probability of getting output v in
state $S_j$]. $\pi$ is the initial distribution of each state.
This model has been taken for the source compression in
\cite{markc} but this solves the asymmetric case of Slepian-Wolf
only. Similar model is also used in \cite{lpdcmc} for the channel
with memory. Some work along the lines of the soure correlation
having Markov structure for the asymmetric case of Slepian-Wolf
problem has also been done in \cite{vernoi}. For the distributed
source coding, the results of the previous section still hold with
$R_{c_1} = 1$ , $R_{c_2} = 1$ and different values of $C_f$ and
$C_b$. Hence, define ($R_a,R_b,C_f,C_b$)-SW distributed code as
the code which returns value of a,b,c and R as the parameters of
the code construction as in previous section taking values of
$R_a,R_b,C_f,C_b$ by the two step approach in equations
\eqref{eq3} and \eqref{eq4}
\begin{equation}\label{eq3}
\begin{gathered}
  (1 - \frac{1} {c})(\frac{1} {R} - 1) =
  (1 - C_f )(1 - \frac{1}
{a}) \hfill \\
  (1 - \frac{1} {b})(\frac{1} {R} - 1) = \frac{1} {a}(1 -  C_b) \hfill \\
  R_a   =  [\frac{1} {a} + \frac{{\frac{1} {R} - 1}}
{b}] \hfill \\
  R_b   =  [1- \frac{1} {a} + \frac{{\frac{1} {R} - 1}}
{c}] \hfill \\
\end{gathered}
\end{equation}

\begin{equation}\label{eq4}
\begin{gathered}
  R' \leftarrow \frac{R} {{R + \max (\frac{{1 - R}} {b},\frac{{1 -
  R}}
{c})}} \hfill \\
  b \leftarrow b \times \frac{{R(1 - R')}}
{{R'(1 - R)}} \hfill \\
  c \leftarrow c \times \frac{{R(1 - R')}}
{{R'(1 - R)}} \hfill \\
R \leftarrow R' \hfill \\
\end{gathered}
\end{equation}

Let us describe the encoding and decoding process in detail

\begin{figure}[h]\begin{center}
 \includegraphics[width=9cm]{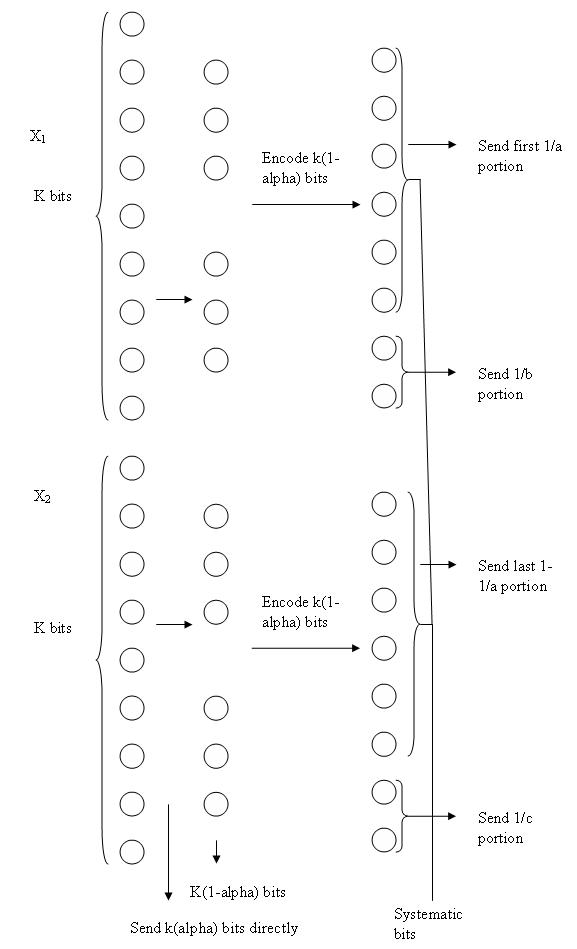}\\
 \caption{Encoding scheme for Markov correlation model}
\label{encmem}\end{center}
\end{figure}

Encoding:

We have two sources $X_1$ and $X_2$ each of k bits. The scheme has
been shown in Figure \ref{encmem}. Send a portion k*alpha of the
bits equally spaced directly to the receiver. From the remaining
k(1-alpha) bits, encode them using (($R_{X_1}$ -alpha)/(1-alpha),
 ($R_{X_2}$ -alpha)/(1-alpha) , $C_f$, $C_b$)-SW Distributed code defined earlier where $R_{X_1}$, $R_{X_2}$, $C_f$ and $C_b$ are respectively
 the desired source rates for the two sources and the forward and backward capacities of the channel between the two sources.

\begin{figure}[h]\begin{center}
 \includegraphics[width=10cm]{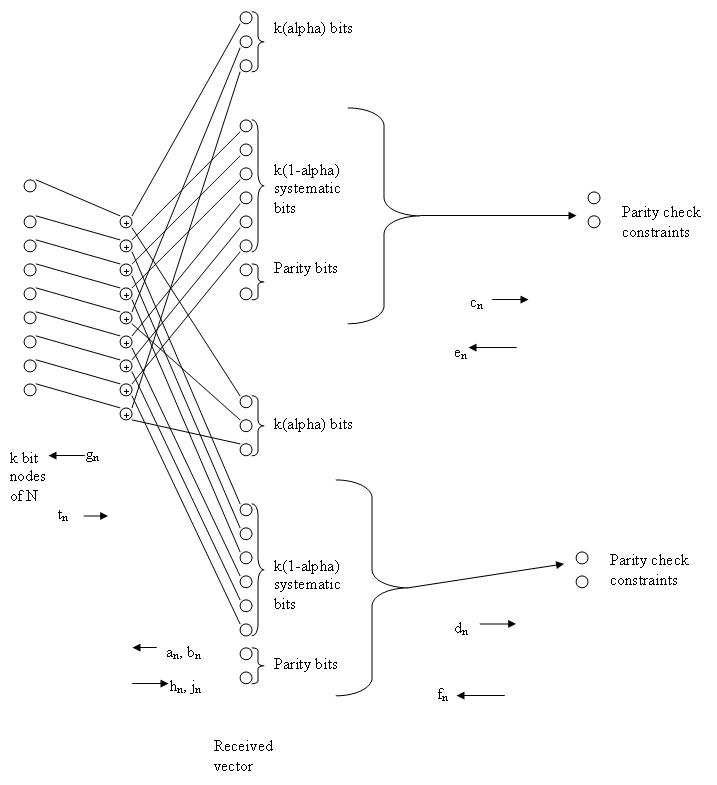}\\
 \caption{Decoding scheme for Markov correlation model}
\label{decmem}\end{center}
\end{figure}
Decoding:

 The decoding approach has been shown in Figure
\ref{decmem}. The steps of message parsings in the scheme are:
\begin{enumerate}
    \item Message Parsing from the bit nodes to the Markov Model:
The systematic bits of the two sources and the k*alpha bits that
were sent without any encoding are used to get the likelihoods of
the bits of N, as N is the sum of the two sources, and hence we
can easily get their likelihoods from the likelihoods of the bits
of the two sources. The likelihoods of the nodes of the received
vector corresponding to the systematic bits of the two sources be
$a_n^x$ and $b_n^x$ respectively with n denoting the $n^{th}$ bit
with n varying from 1 to k and the ordering is according to the
ordering in the input source vector of the two sources, then the
likelihoods of the nodes of N $g_n^x$ are given by $\\
  g_n^0  = a_n^0  \times b_n^0  + a_n^1  \times b_n^1 \\  g_n^1  = a_n^0  \times b_n^1  + a_n^1  \times b_n^0
  $
 \item Message Parsing from Markov Model to the bit nodes: After
receiving the noisy version of N, and knowing the memory structure
of N, we can use trellis decoding to update the likelihoods of N
and send the likelihoods towards the bit nodes. Consider the
trellis for a finite-state binary Markov model. The starting and
ending states associated with a particular edge $'e'$ are
represented by $s^S(e)$ and $s^E(e)$, respectively, and the bit
corresponding to $'e'$ is denoted by $\epsilon(e)$. The trellis
has two parallel branches between states (one associated with the
bit $\epsilon(e) = 0$ and the other with the bit $\epsilon(e) =
1$). Each one of the branches in the trellis will have an
associated a priori probability $a_e$, which is obtained from the
parameters of the Markov Model. The resulting equations that
implement the belief propagation algorithm over the Markov model
are given by\[
\begin{gathered}
  \alpha _k (s) = \sum\limits_{e:s^E (e) = s} {\alpha _{k - 1} [s^S (e)]} a_e g_k^{\epsilon (e)}  \hfill \\
  \beta _k (s) = \sum\limits_{e:s^S (e) = s} {\beta _{k + 1} [s^E (e)]} a_e g_{k + 1}^{\epsilon (e)}  \hfill \\
  t_n^x  = \eta _n \sum\limits_{e:\epsilon (e) = x} {\alpha _{n - 1} [s^S (e)]\beta _n [s^E (e)]} a_e g_n^x  \hfill \\
\end{gathered}
\]
where $\eta _n$ is the normalizing factor so that $t_n^0 + t_n^1 =
1$. Now, this message $t_n^x$ foes back to the systematic portion
of the bit nodes. Using $t_n^x$ we update the less reliable of the
systematic bit nodes of the two sources by the xor operation. The
systematic bits of two sources are updated by messages $h_n^x$ and
$j_n^x$ respectively. Let us consider message update of $h_n^x$,
$j_n^x$ will be similar. The update equations for $h_n^x$ are
given by
\[
\begin{gathered}
  |.5 - a_n^x | > |.5 - b_n^x | \Rightarrow h_n^0  = a_n^0  \times t_n^0  + a_n^1  \times t_n^1 , \hfill \\ b_n^1  = a_n^0  \times t_n^1  + a_n^1  \times t_n^0  \hfill \\
  |.5 - a_n^x | \leqslant |.5 - b_n^x | \Rightarrow h_n^x  = b_n^x  \hfill \\
\end{gathered}
\]
Since one of $a_n$ or $b_n$ will be known perfectly (by the
encoding scheme), this just means to update the other if that is
not known perfectly.

\item Message Parsing from the bit nodes to the check nodes: This
message propagation is from the k(1-alpha) systematic bits and the
parity bits to the check nodes using their respective likelihoods,
likelihoods of the systematic nodes given by $h_n^x$ and $j_n^x$
and those of parity bits the same as the initial likelihoods in
the first iteration and those updated from the previous iteration
from check nodes to bit nodes later. This message parsing is the
same as the standard belief propagation scheme

\item Message Parsing from the check nodes to the bit nodes: Using
the belief propagation, we send the updated likelihoods towards
the bit nodes.

\end{enumerate}

Hence, the overall scheme can be summarized as:
\begin{enumerate}
    \item Message parsing from bit nodes to Markov model and back
    to the bit nodes. \item Message parsing from bit nodes to
    check nodes. \item If the parity equations are not satisfied,
    and iterations are less then maximum-iterations(say 100) go to
    Step 1 after message parsing from the check nodes to the bit
    nodes.
\end{enumerate}

\begin{figure}[h]\begin{center}
 \includegraphics[width=9cm]{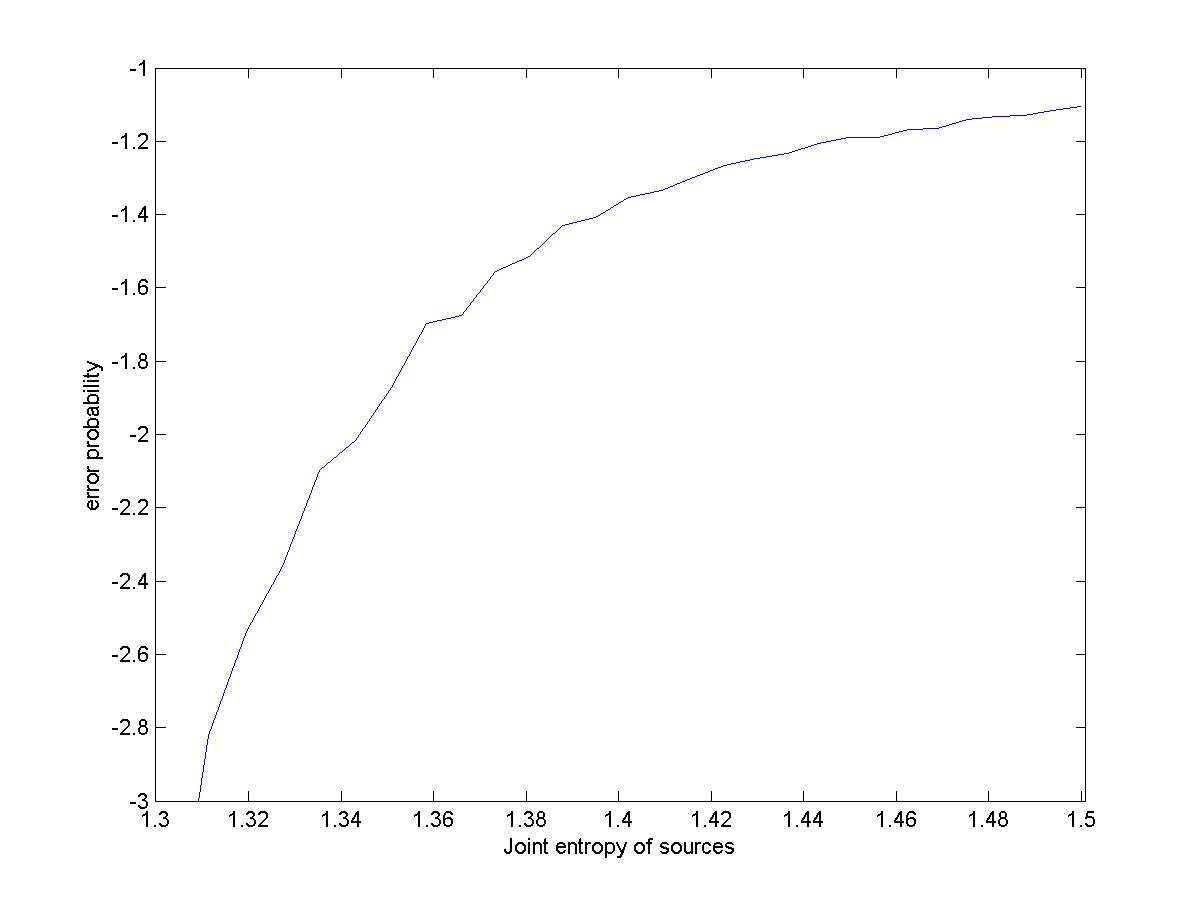}\\
 \caption{Simulation Results } \label{simu7}\end{center}
\end{figure}
\section{Simulations of Distributed Source coding for Markov correlation}
\label{simmar} Take the source correlation model as a Markov
source with two states $S_1$ with probability of 0 as one, and
$S_2$ with probability of 0 as zero, and the transition
probability from $S_1$ to $S_2$ and from $S_2$ and $S_1$ are the
same and equal to p, the parameter which decides the entropy rate
in the plot. Also, take the first source to be i.i.d, making $C_f$
= $C_b$ = 1-h(p).  Take k=6250, alpha=.2, and $R_{X_1}$ =
$R_{X_2}$ = 0.75. Hence, we send 6250*.2=1250 bits of each source
without encoding and for the remaining bits we use the equations
\eqref{eq3} and \eqref{eq4} with $R_{a} =R_{b}$ = .6875, and $C_f$
= $C_b$ = .5, we get a=2, b=1, c=1, R=.8421. The simulation with
varying value of p is shown in Figure \ref{simu7}.

\section{Conclusions}
\label{conc}
 This work deals with the duality between source and
channel coding. We take ideas developed for channel coding and
transform them appropriately to construct joint source channel
coding techniques. Distributed source coding scheme is a special
case of the above work with the capacity of the channels to be
equal to 1. We illustrated how to do joint distributed source
channel coding at arbitrary source rates in the Slepian-Wolf rate
region with arbitrary memoryless source correlation and arbitrary
channel. We  also illustrated how to do distributed source coding
at arbitrary point in the Slepian-Wolf rate region when the source
correlation model has memory. In all the simulations, we
considered regular (3,x) LDPC codes. Taking irregular codes
instead may perform better due to irregularity in the coding
scheme.

Further work can be done to extend the distributed source coding
for sources with Markov correlation to joint source channel
coding. Joint source channel coding for sources with memory has
been studied in \cite{verd}\cite{fria}. Some work on the
asymmetric case for this problem has also been done using turbo
codes in \cite{ser}. But the problem using LDPC codes and doing
for general source rates are still open problems.

\bibliographystyle{IEEEbib}

\begin{thebibliography}{1}
\bibitem{sw} D. Slepian and J. K. Wolf, {\em Noiseless coding of correlated
information sources,} IEEE trans on Information theory, vol. 19,
pp. 471 - 480, July 1973.

\bibitem{discus} S. S. Pradhan and K. Ramchandran, {\em Distributed source coding
using syndromes (DISCUS): design and construction,} Proc. IEEE
Data Compression Conference, pp. 158-167, March 1999.
  \bibitem{swcoop} J. Li and R. Hu, {\em Slepian-Wolf Cooperation: A Practical and
Efficient Compress-and-Forward Relay Scheme} Proceeding of 43rd
Annual Allerton Conference on Communication, Control and Computing
(Allerton) , St. Louis, MO, Nov. 2005.
\bibitem{dca} D. Schonberg, K. Ramachandran, S.S. Pradhan, {\em
Distributed code constructions for the entire Slepian-Wolf rate
region for arbitrarily correlated sources} Proceedings of the Data
Compression Conference, 2004.
\bibitem{appsen} S. S. Pradhan and K. Ramchandran, {\em Distributed source coding: symmetric rates and applications to
sensor networks,} Proc. IEEE Data Compression Conference, pp.
363-372, March 2000.

\bibitem{jscn} J. Garcia-Frias, {\em Joint source-channel decoding of correlated
sources over noisy channels,} Proc. IEEE Data Compression
Conference, pp. 283-292, 2001.

 \bibitem{ct} A. Aaron and B. Girod,
{\em Compression with side information using turbo codes,} Proc.
IEEE Data Compression Conference, pp. 252-261, April 2002.

\bibitem{tsc} P. Mitran and J. Bajcy, {\em Turbo source coding: A noise-robust
approach to data compression,} Proc. IEEE Data Compression
Conference, p. 465, April 2002.

\bibitem{dst} A. D. Liveris, Z. Xiong, and C. N. Georghiades, {\em A distributed
source coding technique for correlated image using turbo codes,}
IEEE Comm. Letters, vol. 6, pp. 379-381, Sept. 2002. \bibitem{cbt}
J. Garcia-Frias and Y. Zhao, {\em Compression of correlated binary
sources using turbo codes,} IEEE Comm. Letters, vol. 5, pp.
417-419, Oct. 2002. \bibitem{ldpa} D. Schonberg, K. Ramchandran,
and S. S. Pradhan, {\em LDPC codes can approach the Slepian-Wolf
bound for general binary sources,} Proc. of fortieth Annual
Allerton Conference, Urbana-Champaign, IL, Oct. 2002.
\bibitem{cld} A. D. Liveris, Z. Xiong, and C. N. Georghiades,
{\em Compression of binary sources with side information at the
decoder using LDPC codes,} IEEE Comm. Letters, vol. 6, pp.
440-442, Oct. 2002.
\bibitem{sarfek} Mina Sartipi and Faramarz Fekri, {\em Distributed source coding in wireless sensor
networks using LDPC coding: The entire Slepian-Wolf Rate Region}
IEEE Communications Society WCNC 2005.
\bibitem{rc}H. Pishro-Nik
and F. Fekri, {\em Results on punctured low-density parity check
codes and improved iterative decoding techniques} Submitted to
IEEE Transactions on Information Theory.
\bibitem{mitz}M. Mitzenmacher,{\em A note on low density parity check codes for erasures
and errors}, SRC Technical Note 1998-017, December 1998.
\bibitem{markc} J.G. Frias and W. Zhong, {\em LDPC codes for
compression of multi-terminal sources with hidden  Markov
correlation}, IEEE Communication Letters, Vol. 7, No. 3, Mar 2003.
  \bibitem{lpdcmc}   J.G. Frias,{\em Decoding of low-density
  parity-check codes over finite-state binary Markov channels},
  IEEE Transactions on Communications, Vol. 52, No. 11, Nov. 2004.
\bibitem{vernoi} Giuseppe Caire, Shlomo Shamai, and Sergio Verdu,{\em Noiseless data compression with low-density parity-check codes},
Advances in Network Information Theory  P. Gupta, G. Kramer and A.
J. van Wijngaarden, Eds., DIMACS Series in Discrete Mathematics
and Theoretical Computer Science, vol. 66, pp. 263-284, American
Mathematical Society, 2004
  \bibitem{live}A. D. Liveris, Z. Xiong and C. N. Georghiades
  ,{\em  Joint source-channel coding of binary sources with side information at the decoder using IRA
  codes}, Multimedia Signal Processing, 2002 IEEE Workshop on
          9-11 Dec. 2002 Page(s):53 - 56
  \bibitem{newsc} R. Hu, R. Viswanathan and Jing Li ,{\em  A new coding scheme for the noisy-channel Slepian-Wolf problem: separate design and joint
  decoding}, GLOBECOM '04. IEEE
             Volume 1,  29 Nov.-3 Dec. 2004 Page(s):51 - 55 Vol.1
  \bibitem{nearsh} J. Garcia-Frias and Y. Zhao , {\em Near-Shannon/Slepian-Wolf performance for unknown correlated sources over AWGN
  channels},  IEEE Transactions on Communications,  Volume 53,  Issue 4,  April 2005 Page(s):555 - 559

\bibitem{fff} W. Zhong and J. Garcia-Frias,{\em LDGM Codes for Channel Coding and Joint
                                               Source-Channel Coding of Correlated
                                               Sources},    EURASIP Journal on Applied Signal Processing 2005:6, 942-953

\bibitem{galla}R. G. Gallager, {\em Low-density parity-check codes}, MIT Press, 1963.
\bibitem{urban}T. J. Richardson and R. L. Urbanke, {\em The capacity of low-density
parity-check codes under message-passing decoding}, IEEE Trans.
Information Theory 47 (Feb. 2001), 599-618.
\bibitem{wyn}A.
Wyner, {\em Recent results in the Shannon theory,} IEEE Trans.
Inform. Theory, vol. 20, pp. 2 - 10, January 1974.

\bibitem{tur}C. Berrou and A. Glavieux, {\em Near optimum error
correcting coding and decoding: turbo-codes,} IEEE Trans.
Communications, vol. 44, pp. 1261-1271, October 1996.
\bibitem{verd}G. Caire, S. Shamai and S. Verdu, {\em Almost-noiseless joint source-channel coding-decoding of sources with
memory}, 5th International ITG Conference on Source and Channel
Coding (SCC), Jan 14-16, 2004 \bibitem{fria}J. Garcia-Frias and J.
D. Villasenor,{\em Joint turbo decoding and estimation of hidden
Markov sources}, IEEE Journal on Selected areas in Communications,
Vol. 19, No. 9, Sept. 2001

\bibitem{ser} J. D. Ser, P. M. Crespo and O. Galdos, {\em Asymmetric joint source-channel coding for correlated
sources with blind HMM estimation at the receiver}, EURASIP
Journal on Wireless Communications and Networking 2005:4, 483-492
  \end{thebibliography}

\end{document}